\documentclass{article}

\PassOptionsToPackage{round, compress}{natbib}
\usepackage[preprint]{neurips_2024}

\usepackage[utf8]{inputenc} 
\usepackage[T1]{fontenc}    
\usepackage{hyperref}       
\usepackage{url}            
\usepackage{booktabs}       
\usepackage{amsfonts}       
\usepackage{nicefrac}       
\usepackage{amsmath}
\usepackage{amsthm}
\usepackage{algorithm}
\usepackage{algpseudocode}
\usepackage{bm}
\usepackage{array}
\usepackage{graphicx}
\usepackage{caption}
\usepackage{subcaption}
\usepackage{microtype}      
\usepackage{xcolor}         

\newtheorem{definition}{Definition}[section] 
\newtheorem{theorem}{Theorem}[section]

\title{Role Play: Learning Adaptive Role-Specific Strategies in Multi-Agent Interactions}

\author{
    Weifan Long\textsuperscript{1}, Wen Wen\textsuperscript{2}, Peng Zhai\textsuperscript{1,3}\thanks{Corresponding Author}, Lihua Zhang\textsuperscript{1,3}\footnotemark[1] \\
  \texttt{wflong21@m.fudan.edu.cn, w\_wen1113@tju.edu.cn},\\
  \texttt{pzhai,lihuazhang@fudan.edu.cn} \\
  \\
  \textsuperscript{1}Academy for Engineering and Technology, Fudan University, Shanghai 200433, China
  \\
  \textsuperscript{2}School of New Media and Communication, Tianjin University, Tianjin 300072, China
  \\
  \textsuperscript{3}Ji Hua Laboratory, Foshan 528251, China
}

\begin{document}

\maketitle

\begin{abstract}
    Zero-shot coordination problem in multi-agent reinforcement learning (MARL), which requires agents to adapt to unseen agents, has attracted increasing attention. Traditional approaches often rely on the Self-Play (SP) framework to generate a diverse set of policies in a policy pool, which serves to improve the generalization capability of the final agent. However, these frameworks may struggle to capture the full spectrum of potential strategies, especially in real-world scenarios that demand agents balance cooperation with competition. In such settings, agents need strategies that can adapt to varying and often conflicting goals.
    Drawing inspiration from Social Value Orientation (SVO)—where individuals maintain stable value orientations during interactions with others—we propose a novel framework called \emph{Role Play} (RP). RP employs role embeddings to transform the challenge of policy diversity into a more manageable diversity of roles. It trains a common policy with role embeddings observation and employ a role predictor to estimate the joint role embeddings of other agents, helping the learning agent adapt to its assigned role. We theoretically prove that an approximate optimal policy can be achieved by optimizing the expected cumulative reward relative to an approximate role-based policy. Experimental results in both cooperative (Overcooked) and mixed-motive games (Harvest, CleanUp) reveal that RP consistently outperforms strong baselines when interacting with unseen agents, highlighting its robustness and adaptability in complex environments.
\end{abstract}

\section{Introduction}
\label{sec1}
Artificial Intelligence (AI) has achieved remarkable success in mastering a wide range of strategic and competitive games, as demonstrated by notable research in this area \citep{silver2018general, vinyals2019grandmaster, berner2019dota}. Significant advancements have also been made in cooperative settings, where agents are trained to collaborate with either humans or other agents to achieve common goals \citep{carroll2019utility, zhao2023maximum}.
However, in real-world applications such as autonomous driving, interactions among agents often display mixed motives, combining elements of both cooperation and competition \citep{schwarting2019social}. In these mixed-motive environments, agents face complex interactions where each participant has distinct objectives. For example, in a public goods game, agents must carefully balance the benefits of contributing to a collective resource against the costs of their individual contributions \citep{gibbons1992primer}. These environments present significant challenges, requiring agents to develop sophisticated adaptation strategies to effectively interact with others who have varying incentives.

Zero-shot coordination is well-recognized in the context of multi-agent reinforcement learning (MARL), particularly for agents that need to interact effectively with new partners they have not encountered during training \citep{hu2020other}.
Self play (SP) is an effective framework for this challenge \citep{anyplay, lupu2021trajectory, zhao2023maximum}. The SP-based frameworks typically build a policy pool through SP, utilized to enhance the generalization capabilities of the final agent. Various techniques aim to increase the diversity of this policy pool to improve the agent's ability to generalize across different scenarios \citep{garnelo2021pick, zhao2023maximum}. In mixed-motive games, there is a greater need for policy diversity to adapt to the varying and often conflicting goals resulting from the imperfect alignment of incentives among group members.
However, the policy pool, primarily composed of past iterations of policies from a given population, captures only a limited range of the policy space. This limitation can prevent agents from effectively managing novel situations or policies not previously encountered in the training set.

Unlike existing policy pool based works, our approach try to develop a general model which can generate policies with different value orientation. 
Given the inherent challenges in representing policies directly due to their complexity, we propose projecting the policy space into a more compact dimension. 
Inspired by Social Value Orientation (SVO) \citep{svo}, in which individuals maintain stable value orientations (roles) in interactions with others, we proposed \emph{Role Play} (RP), which compress the vast MARL policy space into a more manageable ``human role space.''
This simplification aims to improve both the interpretability and efficiency of agent interactions.
Furthermore, drawing on social intuition \citep{lieberman2000intuition, jellema2024social} that humans estimate the behaviors of others during interactions to make better decisions, we introduce a role predictor to estimate the joint role embeddings of other agents, aiding the learning agent in adapting to its assigned role. This setup enables agents to learn and adapt to their assigned roles more effectively, enhancing their performance across various interactive scenarios.

In this work, we introduce a novel framework, \emph{Role Play} (RP), specifically designed to address the zero-shot coordination problem in multi-agent interactions. Our approach is distinguished by several key innovations:
\begin{itemize}
    \item \textbf{Role Embedding:} We utilize a sophisticated reward mapping function to project the extensive policy space into a more manageable role embedding space. This transformation facilitates  structured and strategic navigation through the complex landscape of agent behaviors. We theoretically prove that an approximate optimal policy can be obtained by optimizing the expected cumulative reward with respect to an approximate role-based policy. 
    \item \textbf{Role Predictor:} Inspired by social intuition, we have developed a role predictor that estimates the joint role embeddings of other agents. This module enhances the agent's ability to accurately predict and adapt to the role-based policies of other agents, enabling the learning agent to adapt more effectively to its assigned role.
    \item \textbf{Meta-task Learning:} We employ meta-learning techniques to model agent interactions as meta-tasks, which allows the learning agent to extrapolate from limited experiences to new, unseen scenarios. This approach significantly improve the adaptability of the learning agent to different roles and strategies.
\end{itemize}
These innovations collectively enhance the capability of agents to adapt and perform in complex multi-agent environments, establishing RP as a robust solution to zero-shot coordination challenges in MARL. To gain a deeper understanding of our framework and explore additional visualizations, we invite readers to visit our project website, where more detailed results are provided\footnote{\url{https://weifan408.github.io/role\_play\_web/}}.

\section{Related works}
A significant body of research focuses on enhancing the zero-shot coordination of MARL agents when interacting with unfamiliar partners \citep{kirk2021survey}. Most existing methods that aim to improve generalization across diverse agent interactions utilize the SP framework \citep{lanctot2017unified, bai2020near}, typically employing a policy pool to train agents in adapting to varied strategies. FCP \citep{strouse2021collaborating} introduces a two-stage framework that initially generates a pool of self-play policies and their prior versions, followed by training an adaptive policy against this pool. Additionally, some approaches improves the diversity of within the policy pool to cultivate more robust adaptive policies \citep{garnelo2021pick, zhao2023maximum}. TrajeDi \citep{lupu2021trajectory} applies the Jensen-Shannon Divergence across policies to foster the training of diverse strategies. AnyPlay \citep{anyplay} introduces an auxiliary loss function and intrinsic rewards to aid self-play-based agents in generalizing to more complex scenarios of both intra-algorithm and inter-algorithm cross-play. Furthermore, BRDiv \citep{rahman2023generating} assesses diversity based on the compatibility of teammate policies in terms of returns. HSP \citep{yu2023learning} employs event-based reward shaping to enhance the diversity of the policy pool. While these SP-based methods have achieved success, the reliance on a policy pool can be limiting. The policy pool, primarily composed of past iterations of policies from a given population, captures only a limited range of the policy space. This limitation can prevent agents from effectively managing novel situations or policies not previously encountered during training. Our approach diverges from the conventional use of a policy pool. Instead, we aim to learn a versatile policy capable of generating a spectrum of behaviors through role embeddings. This allows for broader adaptation across various strategic scenarios, enhancing the overall diversity and robustness of policies.

Our approach enhances the learning agent's adaptability by predicting the roles of other agents. Opponent modeling is traditional method which characterizing the behaviors, goals, or beliefs of opponents, enabling agents to adjust their policies to effectively adapt to various opponents. ToMnet \citep{rabinowitz2018machine} aims to equip agents with a theory of mind analogous to that of humans. Similarly, SOM \citep{raileanu2018modeling} adopts a unique approach to theory of mind by utilizing the agent's own policy to infer the opponent's goals. LOLA \citep{foerster2017learning} explicitly considers how an agent's policy might influence future parameter updates of its opponent, effectively anticipating and shaping the opponent's learning trajectory. Meta-MAPG \citep{kim2021policy} integrates meta-learning to develop policies that can adapt to the learning process of opponents. Additionally, M-FOS \citep{lu2022model} employs generic model-free optimization methods to learn meta-policies that are effective in long-horizon opponent shaping. While these methods focus on the learning dynamics of opponents, they tend to pay less attention to enhancing the generalization capabilities across diverse agent interactions. In contrast, our approach emphasizes the role-based adaptation, enabling agents to effectively engage with a wide range of roles and strategies without relying on detailed models of opponent learning dynamics.

We introduce meta-learning techniques to enable our policy leverage prior experiences across different roles, allowing it to quickly adapt to new tasks with minimal data or training \citep{beck2023survey}. There are two canonical meta reinforcement learning algorithms that have been widely adopted in the literature \citep{song2019maml, tack2022meta}, Model-Agnostic Meta-Learning (MAML) \citep{finn2017model} which uses meta-gradients, and Fast RL via Slow RL (RL$^2$) \citep{duan2016rl2}, which uses which uses a history-dependent policy. MAML enables models to quickly adapt to new tasks with only a few gradient updates by optimizing the model's parameters such that they can be fine-tuned efficiently for new tasks, without requiring major adjustments to the model architecture itself. RL$^2$ uses a slower, more deliberate reinforcement learning (RL) process to guide and optimize a faster learning agent; the slow RL agent captures long-term strategies and knowledge, which are then leveraged by the fast RL agent to quickly adapt to new tasks or environments. These methods have been successfully applied in various scenarios \citep{yu2020meta,mitchell2021offline}. In our work, we introduce meta-learning techniques to model the interactions between agents as meta-tasks, enabling the learning agent to generalize from a limited set of experiences to new, unseen agents and roles. This approach enhances the adaptability of the learning agent to different roles and strategies, enabling it to effectively navigate complex multi-agent environments. 

\section{Preliminary}
\label{sec3}
In this work, we study a role-based Markov game, which is represented by $\bm{\mathcal{M}} = \left\langle m, \bm{\mathcal{S}}, \bm{\mathcal{A}}, \bm{\mathcal{P}}, \bm{\mathcal{R}}, \bm{\mathcal{O}}, \psi, \bm{z}, \gamma \right\rangle$. Here, $m$ denotes the number of agents, and $\bm{\mathcal{S}}$ represents a finite set of states.
The joint action space is given by $\bm{\mathcal{A}} = \prod_{i=1}^m A^i$, where $A^i$ is the action space for agent $i$. The transition function $\bm{\mathcal{P}}: \mathcal{S} \times \mathcal{A} \times \mathcal{S} \rightarrow [0, 1]$ defines the probability of transitioning from one state to another given a joint action.
Each agent $i$ has a specific reward function $\mathcal{R}^i: \mathcal{S} \times \mathcal{A} \rightarrow \mathbb{R}$, and the collection of these functions forms the joint reward function $\bm{\mathcal{R}} = \{ \mathcal{R}^i \mid i=1,\ldots,m \}$. The joint observation function, $\mathcal{O} = \{\mathcal{O}^i \mid i=1,\ldots,m \}$, where each $\mathcal{O}^i$ provides observations for agent $i$.
We introduce a joint role embedding is represented as $\bm{z} = \{ z^i \mid i=1, \ldots, n \}$, where each $z^i \in \mathbb{Z}$ is the role embedding for agent $i$, and $\mathbb{Z}$ is the role embedding space. 
The reward feature mapping function $\psi(r^i, z^i)$ processes the reward $r^i$ and the role embedding $z^i$ to generate a new reward value for agent $i$.
The discount factor is $\gamma \in [0, 1]$. Each agent $i$ operates under the policy $\pi(a^i \mid o^i)$, where $a^i$ is the action taken given the observation $o^i$. The joint policy of agent $i$ with role embedding $z^i$ is defined as $\pi(z^i)$, and the joint policy of all agents except agent $i$ is denoted as $\bm{\pi}^{-i}(\bm{z}^{-i})=\prod_{j \neq i} \pi(z^j)$, where $\bm{z}^{-i} = \{ z^j \mid j \neq i \}$ is the joint role embedding of all agents except agent $i$. Similarly, the joint observation of all agents except agent $i$ is $\bm{o}^{-i} = \{ o^j \mid j \neq i \}$, and the joint action of all agents except agent $i$ is $\bm{a}^{-i} = \{ a^j \mid j \neq i \}$.

\begin{figure}[h]
    \centering
    \includegraphics[width=0.6\textwidth]{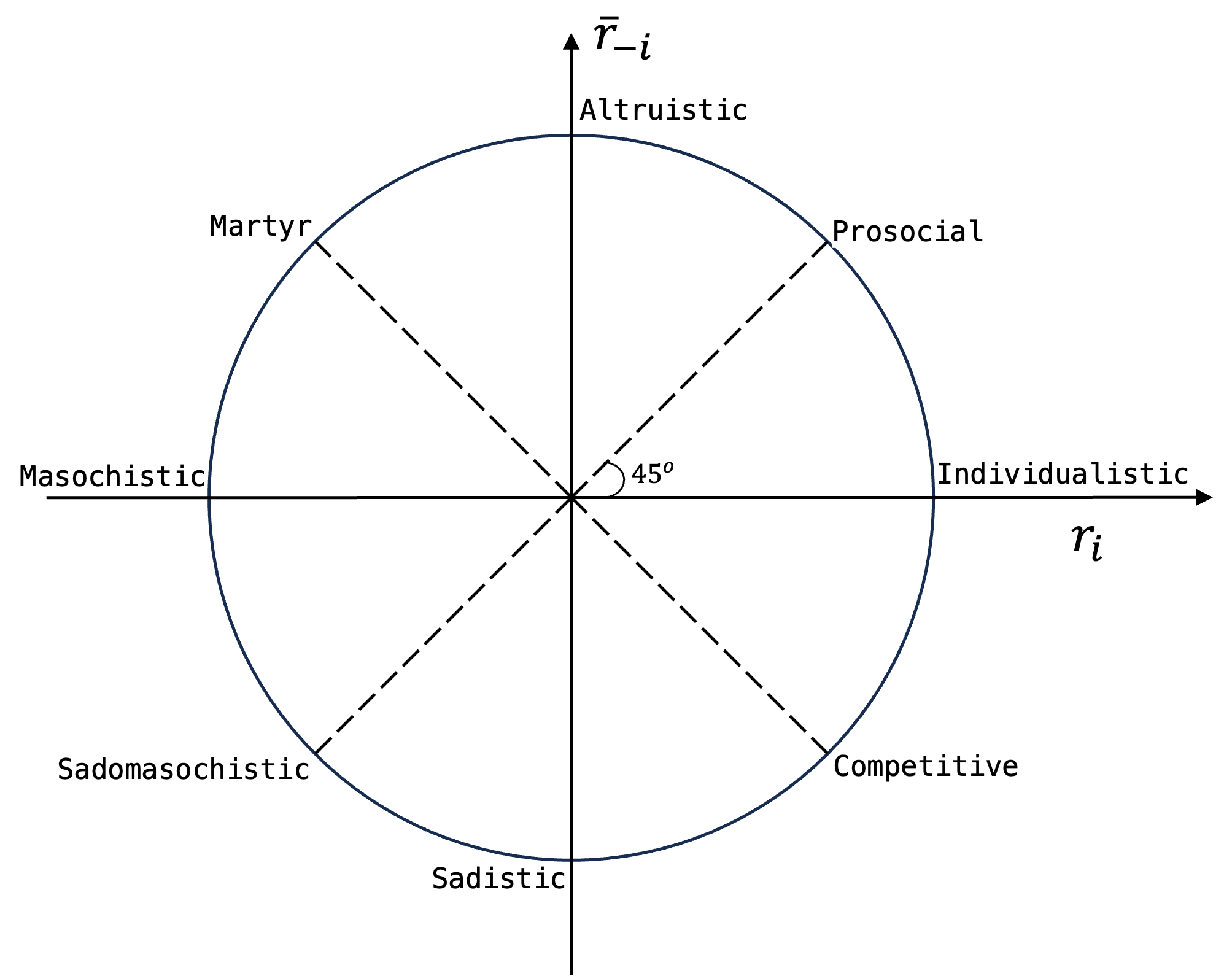} 
    \caption{A graphical representation of the SVO framework.}
    \label{fig_svo}
\end{figure}

We employ Social Value Orientation (SVO) \citep{murphy2014social} to categorize agent roles in mixed-motive games. SVO is a psychological construct that captures individual differences in the value orientations during social interactions. It classifies individuals based on their preferences for outcomes. The ring formulation of SVO is illustrated in Fig. \ref{fig_svo}. The SVO angle of agent $i$, denoted as $\theta^i$, is computed based on the agent's own reward and the average reward of other agents:
\begin{equation}
    \theta^i = arctan\left( \frac{r^i}{\bar{r}^{-i}} \right),
\end{equation}
where $r_i$ is the reward of agent $i$ and $\bar{r}^{-i}$ is the average reward of other agents. In general, SVO can be categorized into eight categorizes: \textit{Masochistic}, \textit{Sadomasochistic}, \textit{Sadistic}, \textit{Competitive}, \textit{Individualistic}, \textit{Prosocial}, \textit{Altruistic} and \textit{Martyr}, corresponding to $\theta^i = \frac{k\pi}{4}$ for $k = -4, \dots, 3$. We use the variant of SVO for reward shaping \citep{peng2021learning}, which defines the SVO-shaped reward for agent $i$ as:
\begin{equation}
    r_i^\prime = cos(\theta^i) \cdot r_i + sin(\theta^i) \cdot \bar{r}^{-i},
\end{equation}
where $r_i^\prime$ is SVO-shaped reward for agent $i$. When $\theta^i = 0$ (\textit{Individualistic}), the reward is the same as the original reward $r^i$. When $\theta^i = \frac{\pi}{4}$ (\textit{Prosocial}), the reward incorporates both the agent's reward and the average reward of others, effectively considering collective outcomes. When $\theta^i = \frac{\pi}{2}$ (\textit{Altruistic}), the reward focuses solely on the average reward of the other agents.

\section{Method}
Role play (RP) is designed to train a single policy which can adapt all roles in role space, enabling agents to dynamically adapt their strategies based on their assigned roles.
Different with existing policy-pool-based methods, RP transforms the policy diversity challenges into more manageable diversity of role representations. The overall framework is illustrated in Fig. \ref{fig_1}.

\begin{figure}[h]
    \centering
    \includegraphics[width=1\textwidth]{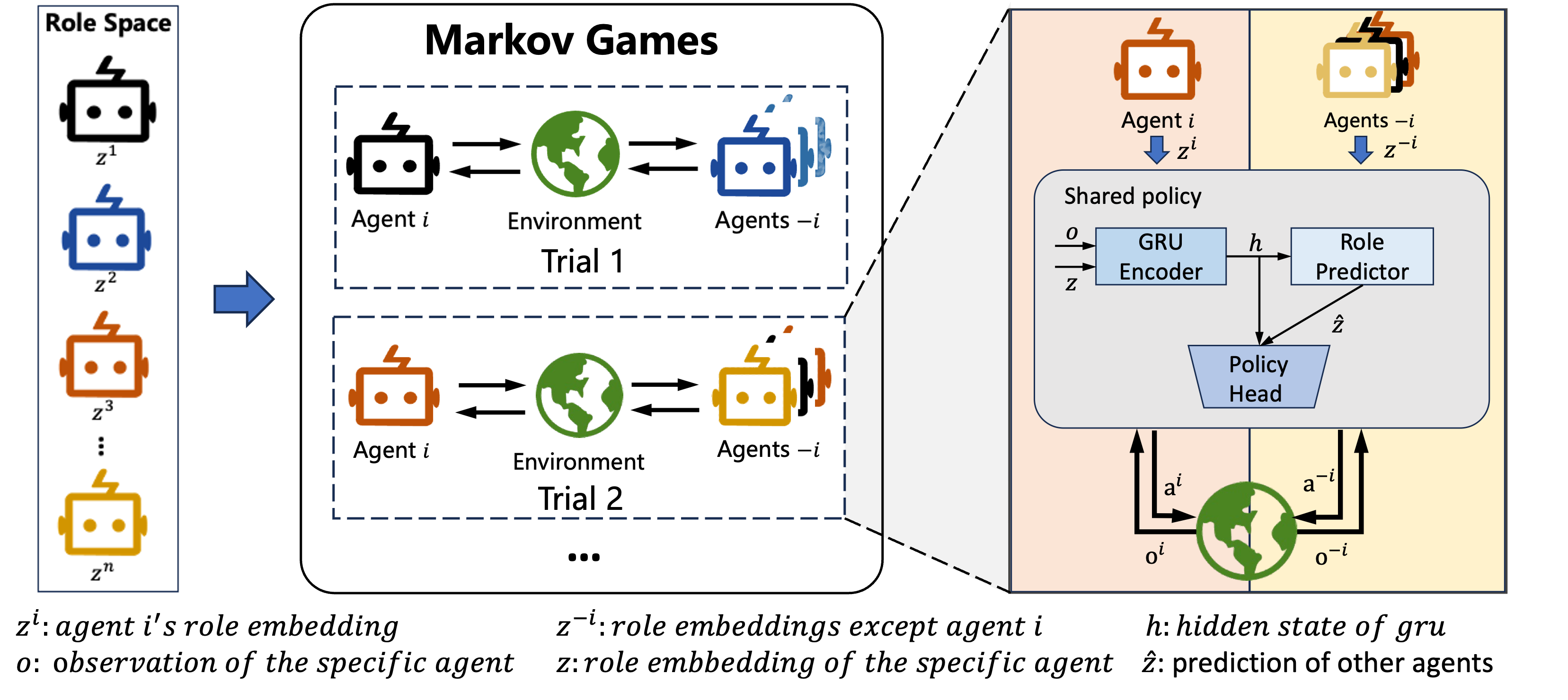} 
    \caption{\textbf{Role play framework}: Different agents share the same policy network but with different role embedding inputs. The role policy prediction network is trained to predict the role embeddings of other agents. The learning agent optimizes its policy based on the predicted roles.}
    \label{fig_1}
\end{figure}

\subsection{Role embedding}
In RP framework, each agent $i$ is assigned a randomly sampled role embedding $z^i$ in every trial and learns to effectively adapt to and play this role. The challenge lies in defining the role space and enabling agents to adapt efficiently to their assigned roles.
To address this, we employ a reward shaping method and introduce a reward feature mapping function $\psi(r^i, z^i)$ that processes the reward $r^i$ and the role embedding $z^i$ to generate a new reward value for agent $i$. This function $\psi(r^i, z^i)$ is specifically designed to capture the role-specific reward information, which is crucial for the agent to adapt to its role.
During training, agents interact with a variety of other randomly sampled roles, optimizing their strategies based on these interactions. Consequently, if the role space is sufficiently diverse to adequately represent the entire policy space adequately, agents are equipped to effectively handle interactions with previously unseen agents.

Each agent $i$ takes action without communications, based on its observation $o^i$ and role embedding $z^i$, receiving individual rewards $r^i$. Its goal is to maximize its expected cumulative discount reward
\begin{equation}
J(\pi) =  \mathop{\mathbb{E}}_{z^i, \bm{z}^{-i} \sim \bm{z}} \left[ J(\pi(z^i), \bm{\pi}(\bm{z}^{-i})) \right],
\end{equation}
where $J(\pi(z^i), \bm{\pi}(\bm{z}^{-i}))$ denotes the expected sum reward achieved by agent $i$ using policy $\pi(z^i)$ while other agents use policies $\bm{\pi}^{-i}(\bm{z}^{-i})$. The detailed form of $J(\pi(z^i), \pi(\bm{z}^{-i}))$ is given by
\begin{equation}
J(\pi(z^i), \pi(\bm{z}^{-i})) =  \mathop{\mathbb{E}}_{\substack{s_{t+1} \sim \bm{\mathcal{P}}(\cdot|s_t, a^i_t, a^{-i}_t)\\ a^i_t \sim \pi^i(\cdot|o^i_t, z^i), \bm{a}^{-i}_t \sim \bm{\pi}^{-i}(\cdot|\bm{o}^{-i}, \bm{z}^{-i})}} \left[ \sum_{t=0}^{\infty} \gamma^t \psi\left(\bm{\mathcal{R}}( s_t,a^i_t,\bm{a}^{-i}_t), z^i \right) \right].
\end{equation}

Given the vast size of the policy space, it is challenging for the mapped role space $\mathbb{Z}$ to fully represent the entire policy space one-to-one.
Inevitably, this reduction results in some loss of information. To address this issue, we analyze the adaptability with respect to missing strategies using an $\epsilon$-close \citep{ko2006mathematical, zhao2023maximum} approximation.

\begin{definition} \citep{zhao2023maximum}
     \label{def_1}
     We define that a random policy $\pi^{\prime}$ is $\epsilon$-close to policy $\pi(z^\prime)$ at observation $o_t$ if 
     \begin{equation}
         \left | \frac{\pi^\prime(a_t^\prime|o_t^\prime)}{\pi(a_t^\prime|o_t^\prime, z^\prime)} - 1 \right | < \epsilon
     \end{equation}
     for all $a_t^\prime \in \mathcal{A}$. If this condition is satisfied at every $o_t^\prime \in \mathcal{O}$, we call $\pi^\prime$ is $\epsilon$-close to $\pi(z^\prime)$.
\end{definition}
Based on Definition \ref{def_1}, if the random policy $\pi^{\prime}$ is $\epsilon$-close to the role policy $\pi(z^\prime)$, we can derive the following theorem.

\begin{theorem}
     \label{thm_1}
     For a finite MDP with $T$ time steps and a specific role policy $\pi(z)$, if any random policy $\pi^\prime$ is $\epsilon$-close to the role policy $\pi(z^\prime)$, then we have
     \begin{equation}
         \left | \frac{J(\pi(z), \pi^\prime)}{J(\pi(z), \pi(z^\prime))} - 1 \right | \leq \epsilon T.
     \end{equation}
\end{theorem}

\begin{proof}
     To maintain the readability of the paper, the proof is included in the \ref{app1}.
\end{proof}

Theorem \ref{thm_1} suggests if a random policy $\pi^\prime$ is $\epsilon$-close to the role policy $\pi(z^\prime)$, the expected cumulative reward for the learning policy $\pi(z)$ in relation to the random policy $\pi^\prime$ is approximately the same as optimizing it in relation to the role policy $\pi(z^\prime)$. This implies that the learning agent can obtain an approximately optimal policy by optimizing the expected cumulative reward with respect to a role-based policy.
However, multi-agent games are complex and dynamic, where each agent's actions can unpredictably influence others. This uncertainty makes it challenging for agents to optimize their behavior effectively.
To address this issue, we propose a prediction method to estimate the $\epsilon$-close role policy.

\subsection{Role predictor}
We aim for the learning agent to accurately predict the joint role embeddings of other agents.
To achieve this, we introduce a role predictor $q_\phi$, which is trained to predict the joint role embeddings of other agents, denoted $\hat{\bm{z}}^{-i}_t$, based on the history of the observation $o_{0:t}$ and its own role embedding $z^i$ at time step $t$. Specifically, the prediction model is formulated as:
\begin{equation}
    \label{eq_5}
    \hat{\bm{z}}^{-i}_t = \arg \max_{\bm{z}^{-i}}  q_\phi(\bm{z}^{-i} | o_{0:t}, z^i).
\end{equation}
This model architecture enables the agent to integrate observational history and predefined role to infer the joint role of other agents within the environment, enhancing its predictive accuracy and adaptability.

Utilizing the role predictor $q_\phi$ and theorem \ref{thm_1}, the learning agent optimizes its policy with the role embedding $z^i$ by maximizing its expected cumulative discount reward 
\begin{equation}
    J(\pi(z^i), \pi(\bm{z}^{-i})) \approx J(\pi(z^i, \hat{\bm{z}}^{-i}))= \mathop{\mathbb{E}}_{\substack{s_{t+1} \sim \bm{\mathcal{P}}(\cdot|s_t, a_t, \bm{a}^{-i}_t)\\ a_t \sim \pi(\cdot|o_t, z^i, \hat{\bm{z}}^{-i})}} \left[ \sum_{t=0}^{\infty} \gamma^t \psi\left(\bm{\mathcal{R}}( s_t,a_t,a^{\prime}_t), z^i \right) \right].
\end{equation}
Once the learning agent is assigned a role, it approximates the dynamics of the MARL problem into a unified single-agent RL framework and optimizes its policy independently.

Up to this point, the learning agent can adaptively interact with any unseen policies for a given role, assuming that both the policy and the role predictor are effectively trained.
Unfortunately, the vast space of joint policies from other agents presents significant challenges in terms of generalization and stability.

\subsection{Modeling role interactions as meta-tasks}
To address the challenges of generalization and stability highlighted in the preceding discussion, we introduce meta-learning techniques \citep{duan2016rl2, finn2017model, fakoor2019meta} into RP. Meta-learning enables agents to generalize from a limited set of experiences to new, unseen scenarios effectively. In our approach, we model the interactions between the learning agent and other agents as meta-tasks, utilizing meta-learning techniques to develop the learning agent's policy.

We treat the role embeddings of other agents as context variables that encapsulate relevant information about the task environment. These role embeddings provide the contextual backdrop for the meta-learning process, offering a rich representation of the dynamic interactions within the multi-agent system. The sampling method of role embedding generates a probability distribution over tasks. Thus, we view the entire learning process as meta-learning tasks, where the agent must adapt to various roles in the environment, guided by the context provided by the role embeddings.

\subsection{Role play framework}
In this subsection, we summarize the RP framework and provide detailed insights into its practical training phase. The algorithm is delineated in Algorithm \ref{alg_1}.

\begin{algorithm}[t]
    \caption{Role Play}
    \label{alg_1}
    \begin{algorithmic}[1]
    \State \textbf{Initialization:}
    \State Initialize role space $\mathbb{Z}$
    \State Initialize reward feature mapping function $\psi$
    \State \textbf{Interaction:}
    \For{each iteration}
        \For{each trail}
            \State Sample role embeddings for all agents ${z^1, \ldots, z^{m}}$
            \For {each time step $t$}
                \For {each agent $i$}
                    \State Predict the joint role embeddings of other agents $\hat{z}^{-i}_t \sim q_\phi(\hat{z}^{-i}_t| o_{0:t},z^i)$
                    \State Sample actions $a^i_t \sim \pi(a^i_t | o^i_t, z^i, \hat{z}^{-i}_t)$
                    \State Observe next observation $o_{t+1}$ and rewards $r_t^i$
                    \State $r_t^i \leftarrow \psi(r_t^i, z^i)$
                \EndFor
            \EndFor
        \EndFor
        \State Update the role predictor $q_\phi$ 
        \State Update the policy $\pi$ by an meta learning algorithm
    \EndFor
    \end{algorithmic}
\end{algorithm}

In practical training, we employ the typically meta-learning algorithm, $RL^2$ \citep{duan2016rl2}, to optimize the policy, as illustrated in Fig. \ref{fig_2}. 
Throughout the iterative process, the algorithm dynamically assigns role embeddings $z^i$ to agents, enabling them to predict joint role embeddings of other agents $\hat{z}^{-i}_t$. This facilitates action selection $a_t^i$ based on these roles and observed interactions. Rewards are subsequently mapped using the defined feature function $\psi$. 
The role predictor $q_\phi$ is updated by minimizing the prediction error between the predicted role embedding and the true role embedding. The learning agent is trained to maximize the expected cumulative mapped reward by interacting with other agents.

\begin{figure}[h]
    \centering
    \includegraphics[width=1\textwidth]{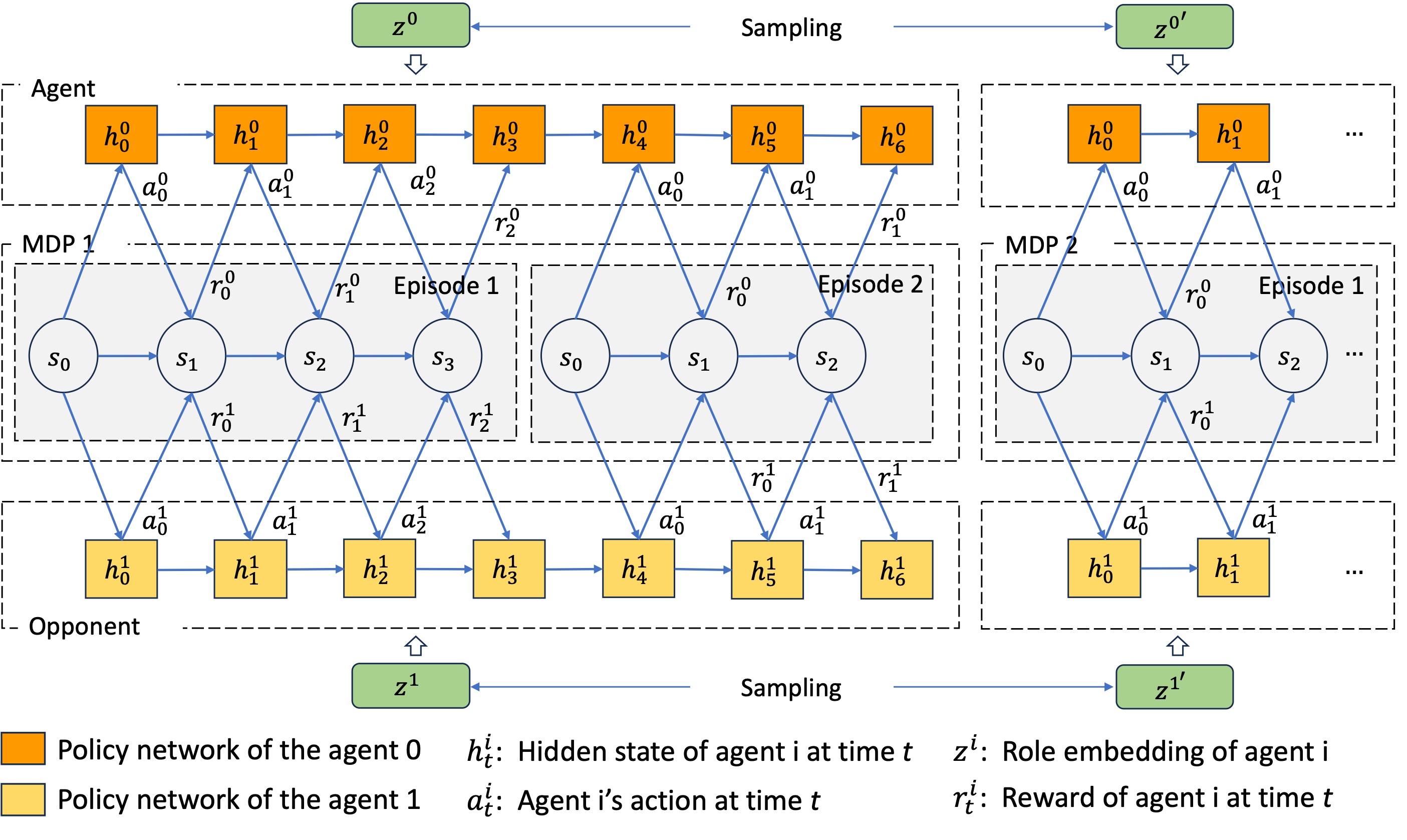} 
    \caption{$RL^2$ based RP}
    \label{fig_2}
\end{figure}

\section{Experiments}
In this study, we initially assess the effectiveness of RP using the cooperative two-player game Overcooked \citep{carroll2019utility}, which is the common benchmark for evaluating the collaboration ability of agents \citep{zhao2023maximum, yu2023learning, rahman2023generating}. We further assess its performance in two-player mixed-motive games, such as Harvest and CleanUp, where agents need to balance cooperation with competition. In these settings, agents must develop strategies that navigate conflicting goals to maximize both individual and collective outcomes. All subsequent internal analysis experiments are conducted within more complex mixed-motive environments. Ablation studies are conducted to verify the effectiveness of the role predictor and the meta-learning methods employed. Role behavior analysis is performed to evaluate the adaptability of the learning agent to different roles. Finally, we analyze performance of the role predictor in predicting the role embeddings of other agents. 

\textbf{Baselines.} We compare the performance of our proposed RP method against several baseline strategies that represent a range of common and state-of-the-art approaches in zero-shot coordination challenge. These include \textit{Trajectory Diversity (TrajeDi)} \citep{lupu2021trajectory}, \textit{AnyPlay} \citep{anyplay}, \textit{Best-Response Diversity (BRDiv)} \citep{rahman2023generating} and \textit{Hidden-Utility Self-Play (HSP)} \citep{yu2023learning}. Each of these methods follows a two-stage process: generating a policy pool and then training a final policy. Notably, while all baseline methods train 16 policies in the first stage to ensure sufficient diversity and adaptability, \textbf{RP achieves comparable or superior performance by training a single policy in a streamlined, single-stage process}, significantly reducing both computational overhead and complexity.

\subsection{Zero Shot Results on Cooperative Games}
Overcooked \citep{carroll2019utility} is a two-player cooperative game intended to test the collaboration ability of agents. In this game, agents work collaboratively to fulfill soup orders using ingredients like onions and tomatoes. Agents can move and interact with items, such as grabbing or serving soups, based on the game state. To complete an order, agents must combine the correct ingredients in a pot, cook them for a specified time, and then serve the soup with a dish to earn rewards. Each order has distinct cooking times and rewards.

\textbf{Experimental setup.} In Overcooked, we employ an event-based reward map function, which is a sparse reward function that rewards agents for some events happening. We select certain events already implemented in Overcooked and assign three preferences levels-hate, neutral and like-for each event which will lead to a negative, zero and positive reward. The reward shaping function is defined as $\psi(r^i, z^i) = r^i + \sum_{k} (z^i_k * E_k)$, where $z^i_k\in[-1,0,1]$ is the preference of event $k$, and $E_k$ is the reward of event $k$. Details of the experimental settings in Overcooked are provided in Appendix \ref{app2_1}.

\begin{figure}[t]
    \centering
    \includegraphics[width=0.99\textwidth]{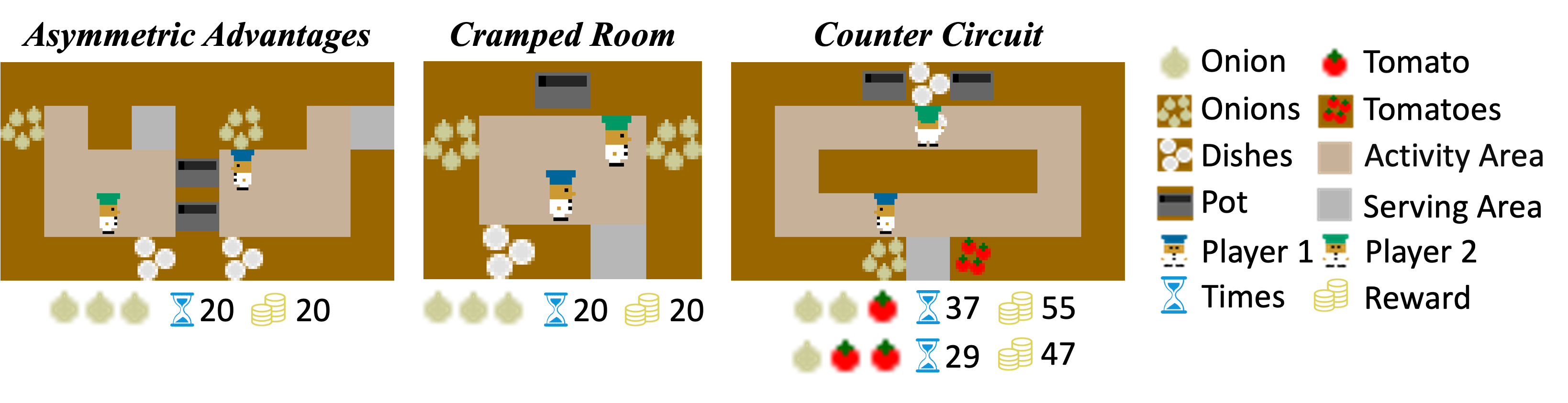}
    \caption{Layouts in Overcooked.}\label{layout}
\end{figure}

Fig. \ref{layout} illustrates the three layouts considered: \textit{Asymmetric Advantages}, \textit{Cramped Room}, and \textit{Counter Circuit}. The \textit{Asymmetric Advantages} layout includes only onions, with two agents separated into distinct areas. \textit{Cramped Room} also features an onion-only layout but places both agents in a shared activity space. \textit{Counter Circuit} presents a more complex setup, with both onions and tomatoes; agents must collect ingredients and prepare the correct soup to earn rewards.

\begin{figure}[h]
    \centering
    \begin{subfigure}[b]{0.49\textwidth}
        \centering
        \includegraphics[width=\linewidth]{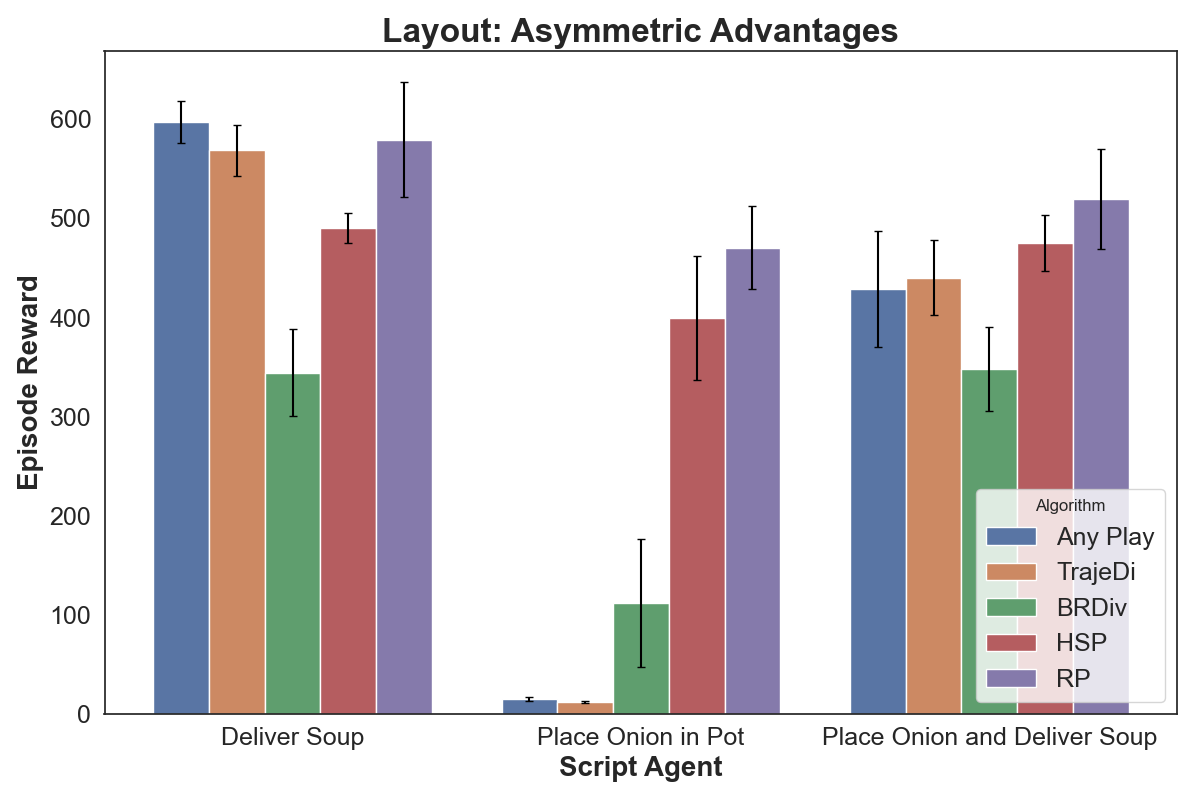}
    \end{subfigure}
    \hfill
    \begin{subfigure}[b]{0.49\textwidth}
        \centering
        \includegraphics[width=\linewidth]{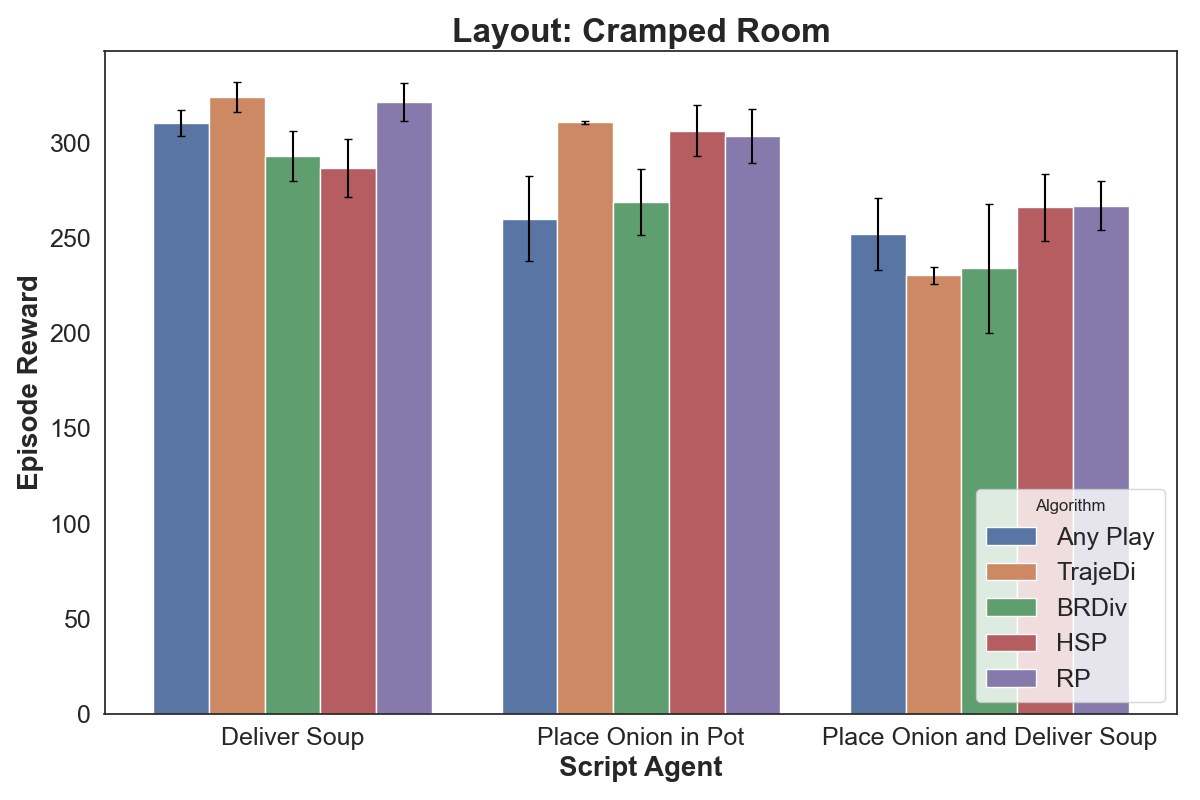}       
    \end{subfigure}
    \caption{Zero-shot evaluation results with script agents on \textit{Asymmetric Advantages} and \textit{Cramped Room}.}
    \label{fig_overcooked1}
\end{figure}

We evaluated the performance of all algorithms in tasks where the trained agent interacts with various scripted agents implemented using HSP \citep{yu2023learning}. In each task, policies engaged with the scripted agent over 100 episodes, and we calculated the mean and standard deviation of the episode rewards. Since the \textit{Asymmetric Advantages} and \textit{Cramped Room} layouts contain only onions, we assessed performance across the following tasks: \textit{Deliver Soup Agent}, which the scripted agent only delivers soup; \textit{Place Onion in Pot Agent}, which the agent only places the onion in the pot; and \textit{Place Onion in Pot and Deliver Soup Agent}, which the agent performs both actions. Fig. \ref{fig_overcooked1} shows the evaluation results for these onion-only layouts. In the \textit{Asymmetric Advantages} layout, AnyPlay and TrajeDi perform well in \textit{Deliver Soup Agent} task but struggle in the \textit{Place Onion in Pot Agent} task, as their policy pools from stage one rarely contain potting policies, which limits their adaptability. By contrast, both HSP and RP overcome this issue using reward shaping methods. While HSP is limited by the size of its policy pool, RP achieves higher rewards by generating a wider range of policies through role embeddings, allowing it to adaptively apply different reward shaping functions. In the simpler layout, \textit{Cramped Room}, all algorithms achieve satisfactory performance. However, RP exhibits a higher level of adaptability by consistently adjusting its policies to suit varying tasks, outperforming others in maintaining robust results across different scripted agent interactions. 

\begin{figure}[h]
    \centering
    \includegraphics[width=\linewidth]{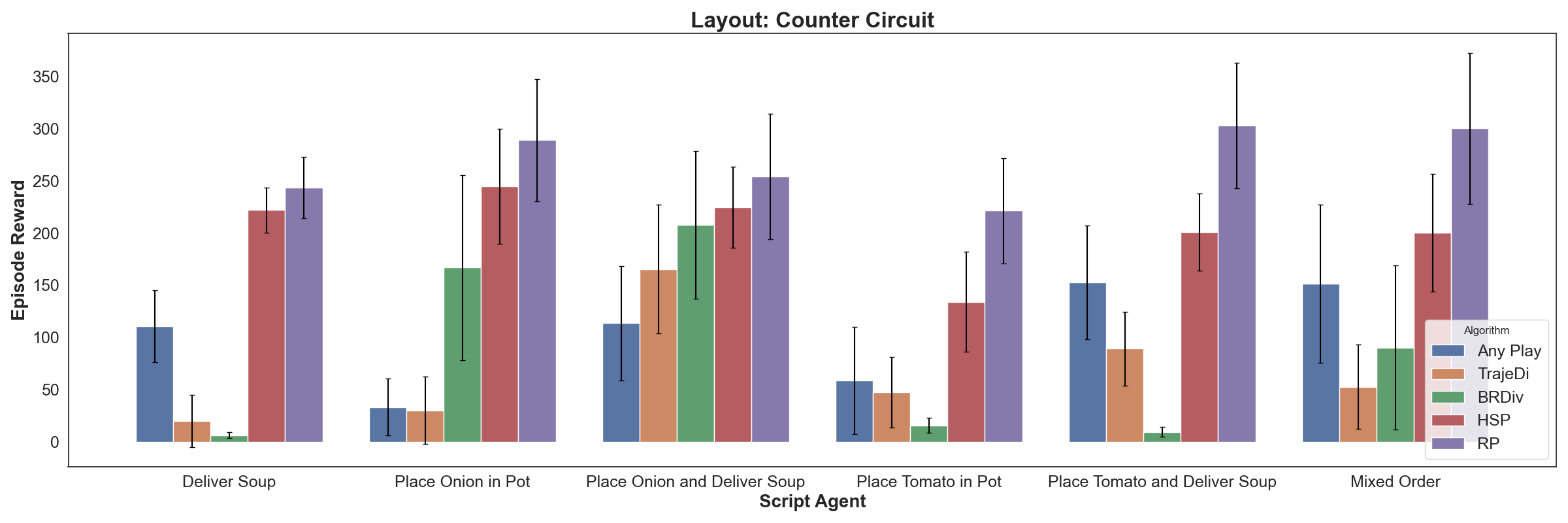}       
    \caption{Zero-shot evaluation results with script agents on \textit{Counter Circuit}.}
    \label{fig_overcooked2}
\end{figure}

In the more complex \textit{Counter Circuit} layout, where agents must gather onions and tomatoes to prepare correct soup, we evaluated performance across three additional tasks: \textit{Place Tomato in Pot Agent}, where the agent only places tomatoes in the pot; \textit{Place Tomato in Pot and Deliver Soup Agent}, where the agent both places tomatoes in the pot and delivers soup; and \textit{Mixed Order}, where the agent places various ingredients in the pot. The results is shown in Fig. \ref{fig_overcooked2}.
AnyPlay, Trajedi and BRDiv obtain policy pools through an implicit way, which limits their ability to cover the diversity of strategies required in complex scenarios, resulting in poor performance across all tasks. In contrast, RP and HSP apply reward shaping methods (an explicit approach) to obtain a broader range of policies, leading to higher rewards. RP consistently outperforms HSP, highlighting its superior adaptability in challenging scenarios. Notably, some results exhibit significant standard deviations, due to the sparse reward structure, where collecting or delivering orders variably impacts the total score.

\subsection{Zero Shot Results on Mixed-Motive Games}
In this subsection, we evaluate the performance of the RP framework in two-player mixed-motive games, \textit{Harvest} and \textit{CleanUp} \citep{hughes2018inequity, leibo2021meltingpot}, which require agents to balance individual contributions against collective benefits. These games introduce a higher level of complexity due to inherent conflicting interests, providing a challenging test environment for agents. To clarify the mixed-motive game settings, we provide detailed explanations of these games.

\textit{Harvest}: agents face a \textit{common-pool resource dilemma} where apple resources regenerate more effectively when left unharvested in groups. Agents can choose to harvest aggressively, risking future availability, or sustainably, by leaving some apples to promote long-term regeneration. Additionally, agents can use beams to penalize others, causing the penalized agent to temporarily disappear from the game for a few steps. 

\textit{CleanUp}: a \textit{public goods game} in which apple growth in an orchard is hindered by rising pollution levels in a nearby river. When pollution is high, apple growth stops entirely. Agents can reduce pollution by leaving the orchard to work in polluted areas, highlighting the importance of individual efforts in maintaining shared resources. Agents are also able to use beams to penalize others. 

\begin{figure}[h]
    \centering
    \includegraphics[width=\linewidth]{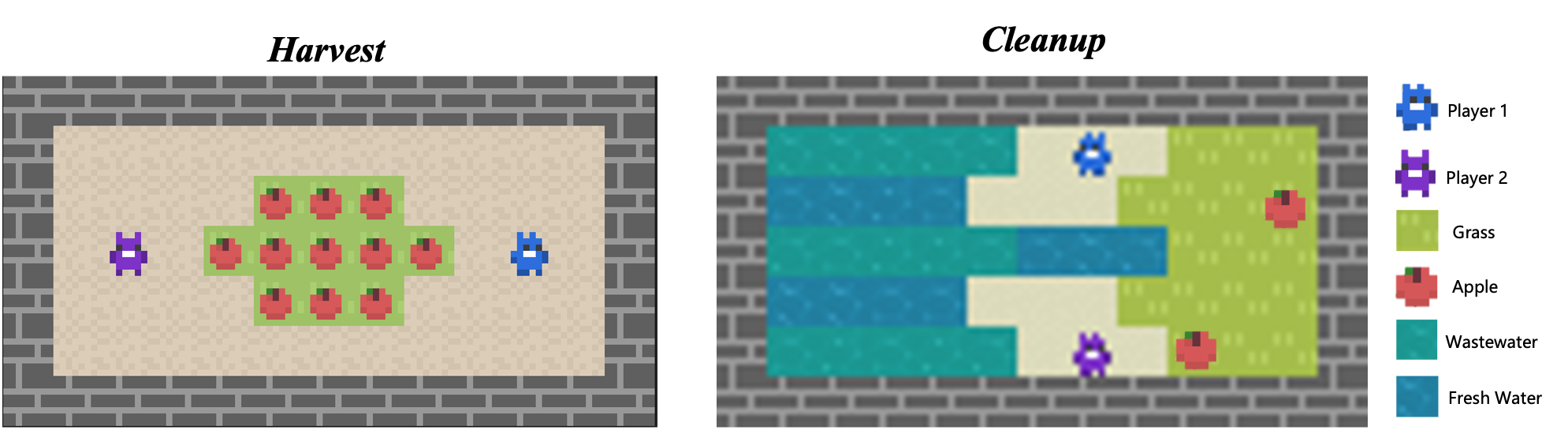}       
    \caption{Harvest and CleanUp games.}
    \label{fig_mp}
\end{figure}

\textbf{Experimental setup.} We implemented the scenarios using the MeltingPot framework \citep{leibo2021meltingpot} and applied SVO \citep{murphy2014social} to categorize agent roles. Utilizing the ring formulation of SVO, we introduced SVO-based (role-based) rewards as intrinsic motivators for agents to learn role-specific strategies. The reward feature mapping function $\psi$ is defined as:
\begin{equation}
\psi(\bm{r}, z^i) = w * r^i + (1-w) * \underbrace{\left| \cos(z^i) r^i + \sin(z^i) \bar{r}^{-i} \right|}_{\text{SVO-based reward shaping}},
\end{equation}
where $r^i$ denotes the individual reward for agent $i$, $\bar{r}^{-i}$ the average reward of other agents, and $w$ a hyperparameter adjusting the impact of role-based rewards. We use eight main roles $\left\{ z^i = \frac{k\pi}{4} \mid k = -4, \dots, 3 \right\}$ for training, as introduced in \ref{sec3}.

In this experiment, we evaluate performance across three tasks for each game. Each task involves a policy pretrained with a specific reward function, with the learning agent interacting with this pretrained agent to assess performance. The tasks are: 
\begin{itemize}
    \item \textit{Selfish Agent}: the agent pretrained with its own reward function.
    \item \textit{Prosocial Agent}: the agent pretrained with the collective reward function.
    \item \textit{Inequity-Averse Agent}: the agent pretrained with the inequity-averse reward function, as described in \citep{hughes2018inequity}.
\end{itemize}
Since the baseline algorithms were not designed for mixed-motive games, we tested them in a cooperative setting by training with collective rewards and evaluating collective rewards in both Harvest and CleanUp to ensure fairness. Further details on baseline algorithms and experimental settings in these mixed-motive games are provided in \ref{app2_2}.

As shown in Table \ref{tab_1}, RP demonstrates comparable performance in both Harvest and CleanUp games, especially in CleanUp game. In Harvest, although AnyPlay and HSP ultimately learn a non-confrontational policy, they contribute little to the collective reward. BRDiv learns a resource-competing policy, leading to resource depletion in the \textit{Selfish Agent} task, and performs moderately in the \textit{Prosocial Agent} task. In the \textit{Inequity-Averse Agent} task, all baselines perform poorly. Instead, RP shows good performance in all tasks, because RP learned to dodge and harvest in interactions with aggressive roles. When RP play an individualistic role ($z=0$), it tends to use the beam to penalize other agents to get a individual reward, which lowers the collective reward. 

In CleanUp, all baselines failed to learn a robust policy that achieves high collective rewards in evaluation tasks, as the environment requires agents to divide the work and cooperate to achieve a high collective outcomes. The policy pools they generated in the first training stage are insufficient to help them learn effective policies. In contrast, RP can utilize the exploration data with cooperative roles to learn a good policy. In the \textit{Selfish Agent} task, RP learns to clean pollution all the time. In the \textit{Prosocial Agent} task, RP learns to clean the pollution while the other agent is harvesting, once the other agent leaves the orchard, RP will harvest apples to maximize collective reward. Since the pretrained inequity-averse agent has lower advantageous inequity aversion\footnote{Details are shown in \ref{app2_2}}, it converges to harvest apples if they are available in the orchard. Therefore, in the \textit{Inequity-Averse Agent} task, RP learns to clean the pollution continuously. We prefer readers to our project website\footnote{https://weifan408.github.io/role\_play\_web/} for visualizations of the results.

\begin{table}
    \centering
    \caption{Collective reward of different algorithms in the Harvest and CleanUp.}
    \label{tab_1}
    \resizebox{\linewidth}{!}{
    \begin{tabular}{ccccccc} 
        \toprule
         & \multicolumn{3}{c}{Harvest} & \multicolumn{3}{c}{CleanUp} \\ 
        \cmidrule(r){2-4} \cmidrule(r){5-7} 
        Algorithm & selfish & prosocial & inequity-averse & selfish & prosocial & inequity-averse \\ 
        \midrule 
        AnyPlay & \textbf{20.62 (4.91)} & \textbf{25.13 (5.29)} & 3.19 (9.12) & 11.58 (8.41)  & 2.42 (12.16) & -1.90 (4.21) \\ 
        Trajedi & 19.74 (5.21) & 23.95 (5.96) &  3.07 (11.06) & -3.41(20.87)  & -23.00(16.49) & -22.30 (13.05) \\
        BRDiv & 12.45 (8.73) & 16.46(5.57) &  -26.97 (17.13) & 8.56 (10.36) & 2.03(17.47) & -5.38 (9.41) \\
        HSP & \textbf{21.71 (8.01)} & 23.24(14.23) &  4.95 (10.95) & 10.90 (9.04) & 8.15(11.86) & -1.20(3.98)\\        
        RP (z=$0$) & 6.94 (12.71) & 13.28 (14.06) & 6.35 (5.82) & 13.23 (11.05) & 24.37 (13.45) & 8.55 (11.16) \\ 
        RP (z=$\frac{\pi}{4}$) & \textbf{20.46 (3.92)} & 23.57 (5.56) & \textbf{15.15 (4.35)} & \textbf{23.35 (13.93)}	 & \textbf{38.09 (7.71)	} & \textbf{23.42 (6.65)} \\
        RP (z=$\frac{\pi}{2}$) &  \textbf{20.60 (3.52)} &  \textbf{25.84 (3.98)} & \textbf{15.61 (6.94)} & \textbf{22.73 (12.58))} &  \textbf{39.53 (6.99)} &  18.64 (10.65) \\
        \bottomrule
    \end{tabular}
    }
\end{table}

\begin{figure}[h]
    \centering
    \begin{subfigure}[b]{0.49\textwidth}
        \centering
        \includegraphics[width=\linewidth]{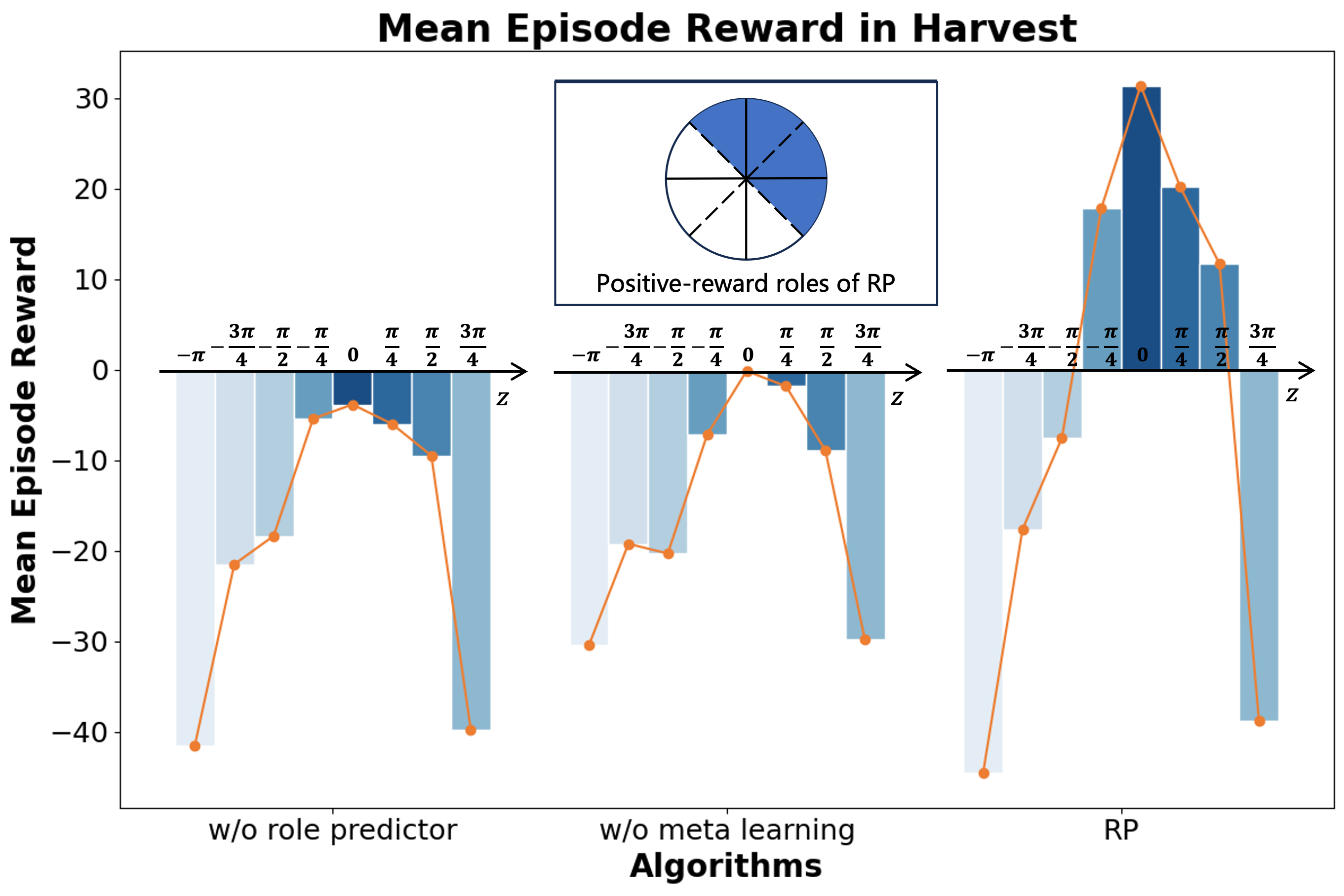}
    \end{subfigure}
    \hfill
    \begin{subfigure}[b]{0.49\textwidth}
        \centering
        \includegraphics[width=\linewidth]{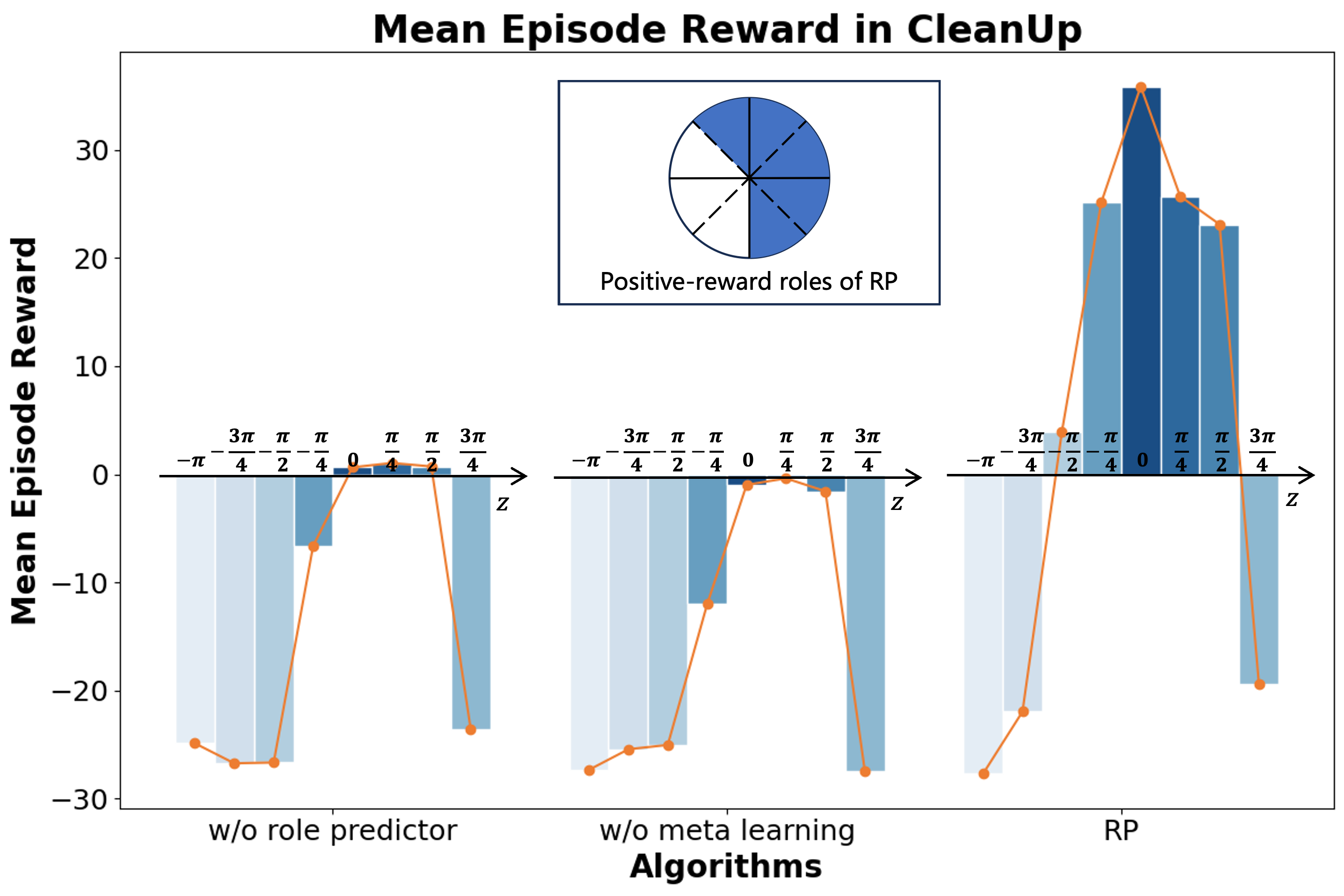}       
    \end{subfigure}
    \caption{Impact of Role Predictor and Meta-Learning on Performance in Harvest and CleanUp.}
    \label{fig_ablation}
\end{figure}

\subsection{Ablation Studies}
\label{ablation}
In this subsection, we examine the impact of the role predictor and the meta-learning method on the performance of the RP framework through comprehensive ablation studies. We assess the contributions of each component by systematically removing the role predictor and the meta-learning method. We analyze interactions by having each role engage separately with all eight distinct roles. For each unique role pairing, agents interact over a series of 1000 episodes. Given the considerable variation in rewards associated with different roles, we employ the mean episode reward of each role as a metric for assessing performance. The results are shown in Fig. \ref{fig_ablation}.
The figure clearly illustrates RP's superior adaptability to diverse roles compared to configurations lacking the role predictor and meta-learning method. The inset provides a detailed view of RP's roles that achieve positive average rewards. Roles commonly seen in real-life settings—such as Altruistic, Prosocial, Individualistic, and Competitive—yield positive rewards. In CleanUp, the Sadistic role also achieves a small positive reward due to its exploitative tendency, though this behavior lacks a virtuous cycle, resulting in only minimal positive gains in average. For the Individualistic role ($z=0$), RP achieves the highest positive rewards. By contrast, in the ablation configuration without the role predictor and meta-learning method, the rewards are significantly lower even negative. Notably, RP attains positive average rewards when adopting prosocial roles and incurs negative rewards with antisocial roles, whereas policies from the ablation studies consistently produce either negative rewards or considerably lower positive rewards across all roles.

\subsection{Role Behavior Analysis}
In our analysis of the RP framework, we investigated the principal role behaviors exhibited in the Harvest and CleanUp games. Experimental setting is the same with \ref{ablation}.
Fig. \ref{fig_role_analysis} visualizes our findings for four main roles that reflect real-life behavioral patterns: Competitive ($z=-\frac{\pi}{4}$), Individualistic ($z=0$), Prosocial ($z=\frac{\pi}{4}$), and Altruistic ($z=\frac{\pi}{2}$). 

The Individualistic role, marked by a significantly larger area on the radar chart, suggests an advantageous position in interpersonal dynamics, a phenomenon supported by findings in individualism studies \citep{triandis2018individualism}. In contrast, the Competitive role, adept at penalizing others and monopolizing resources such as apples, frequently triggers retaliatory actions from other agents, often culminating in considerable negative rewards. The radar chart shows a large number of blank spaces in the lower half, indicating frequent failures in interactions with antisocial roles.

In the more competitive Harvest game, the Altruistic role demonstrates vulnerabilities as it primarily focuses on the welfare of others, often to its own detriment. Meanwhile, the Prosocial role, which aims to harmonize self-interest with the collective good, shows robust performance in this environments by fostering optimal group outcomes.

Conversely, in the CleanUp game, the Prosocial and Altruistic roles engage in environmentally beneficial behaviors, such as pollution removal. While these actions yield lower direct rewards for themselves, they facilitate greater social welfare for other agents. However, roles driven by self-interest, such as the Competitive and Individualistic roles, often reap benefits from these altruistic actions without bearing any of the associated costs. Additionally, when interacting with the more self-sacrificing Martyr role, both the Prosocial and Altruistic roles experience enhanced social welfare. 

Overall, the role behavior analysis demonstrates the adaptability of the RP framework across diverse roles, with each role maintaining its distinct SVO.

\begin{figure}[h]
    \centering
    \begin{subfigure}[b]{0.49\textwidth}
        \centering
        \includegraphics[width=\linewidth]{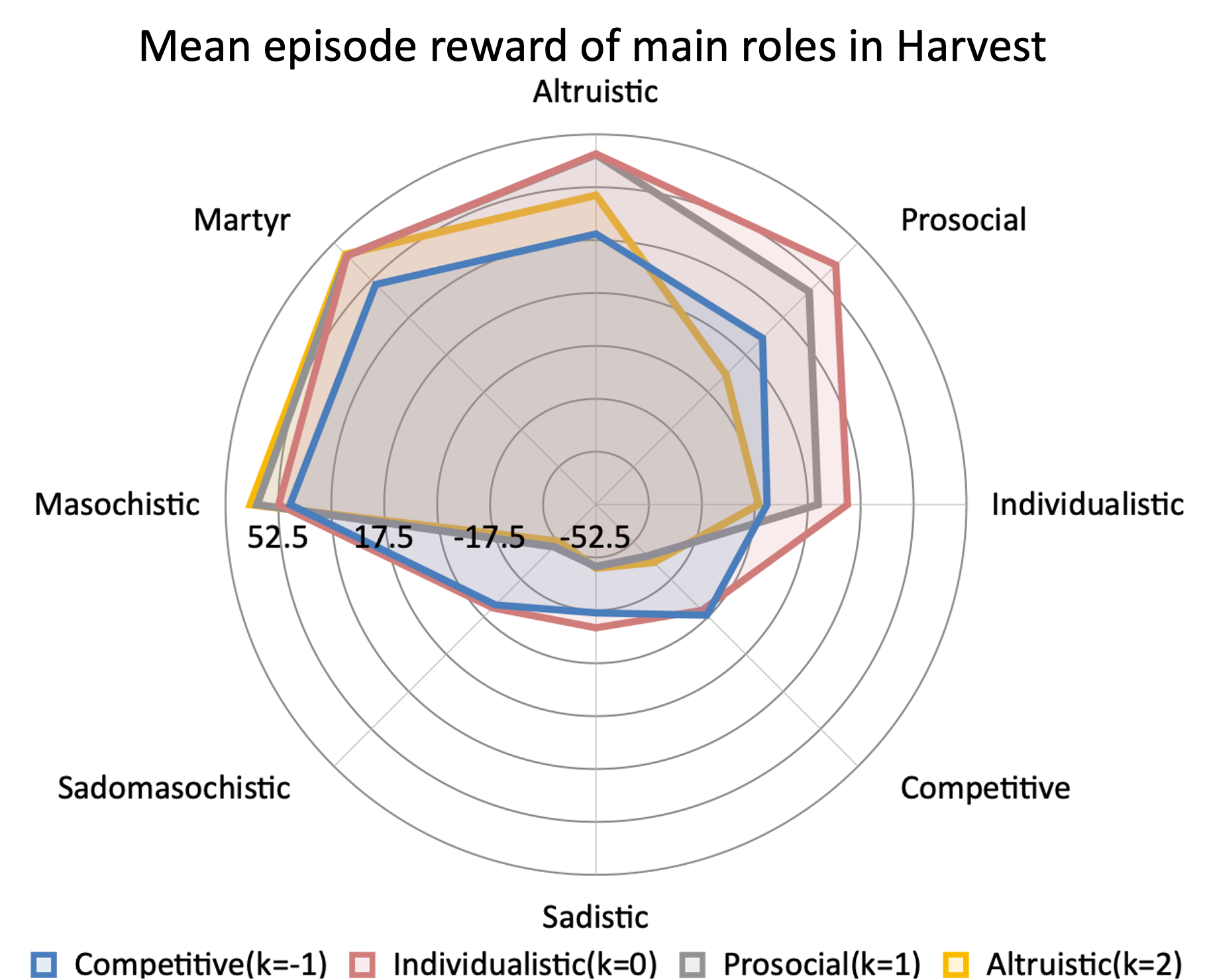}
    \end{subfigure}
    \hfill
    \begin{subfigure}[b]{0.49\textwidth}
        \centering
        \includegraphics[width=\linewidth]{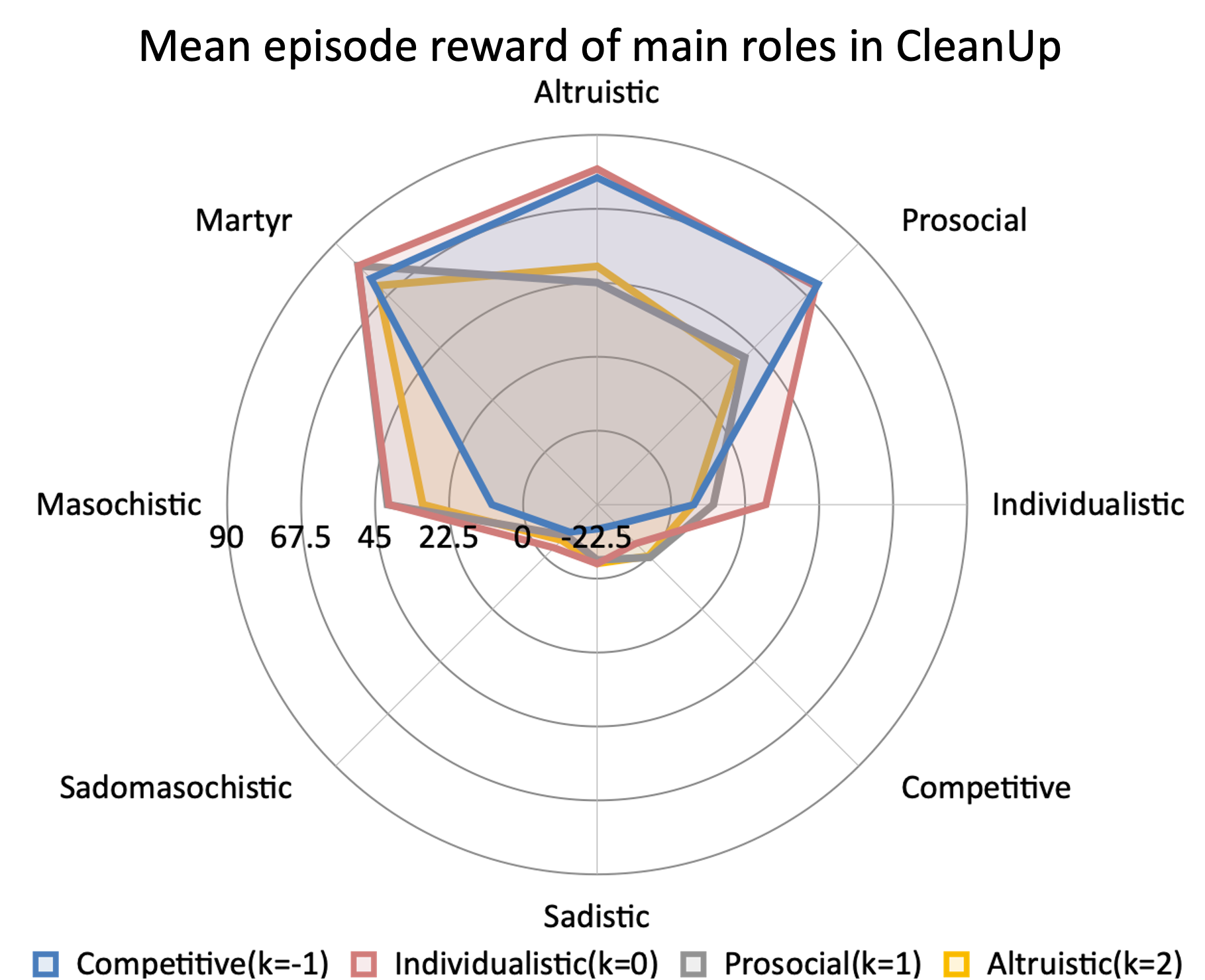}       
    \end{subfigure}
    \caption{Role analysis in Harvest and CleanUp.}
    \label{fig_role_analysis}
\end{figure}

\subsection{Prediction Analysis}
In the RP framework, the role predictor is utilized to estimate the joint role embeddings of other agents. Fig. \ref{fig_pred} displays the prediction outcomes of the role predictor $q(\phi)$ within both the Harvest and CleanUp games. Each heatmap represents the prediction results of a role interacting with all roles. The colors in the heatmap indicate the proportion of predictions that align with actual roles, with colors closer to the diagonal representing higher accuracy. This visualization method clearly illustrates that the most accurate predictions are those where the color intensity peaks along the diagonal, indicating a strong alignment between predicted and actual roles. These results confirm the efficacy of the role predictor in accurately forecasting the role embeddings of other agents. Notably, the predicted roles align with the actual roles or their close counterparts. Although prediction performance is slightly reduced in the CleanUp game due to its more complex environmental dynamics compared to the Harvest game, the role predictor occasionally outputs incorrect roles that exhibit similar behaviors to the actual roles, which still aids the learning agent in adapting to the environment.
 
\begin{figure}[ht]
    \centering
    \begin{subfigure}[b]{\textwidth}
        \centering
        \includegraphics[width=\textwidth]{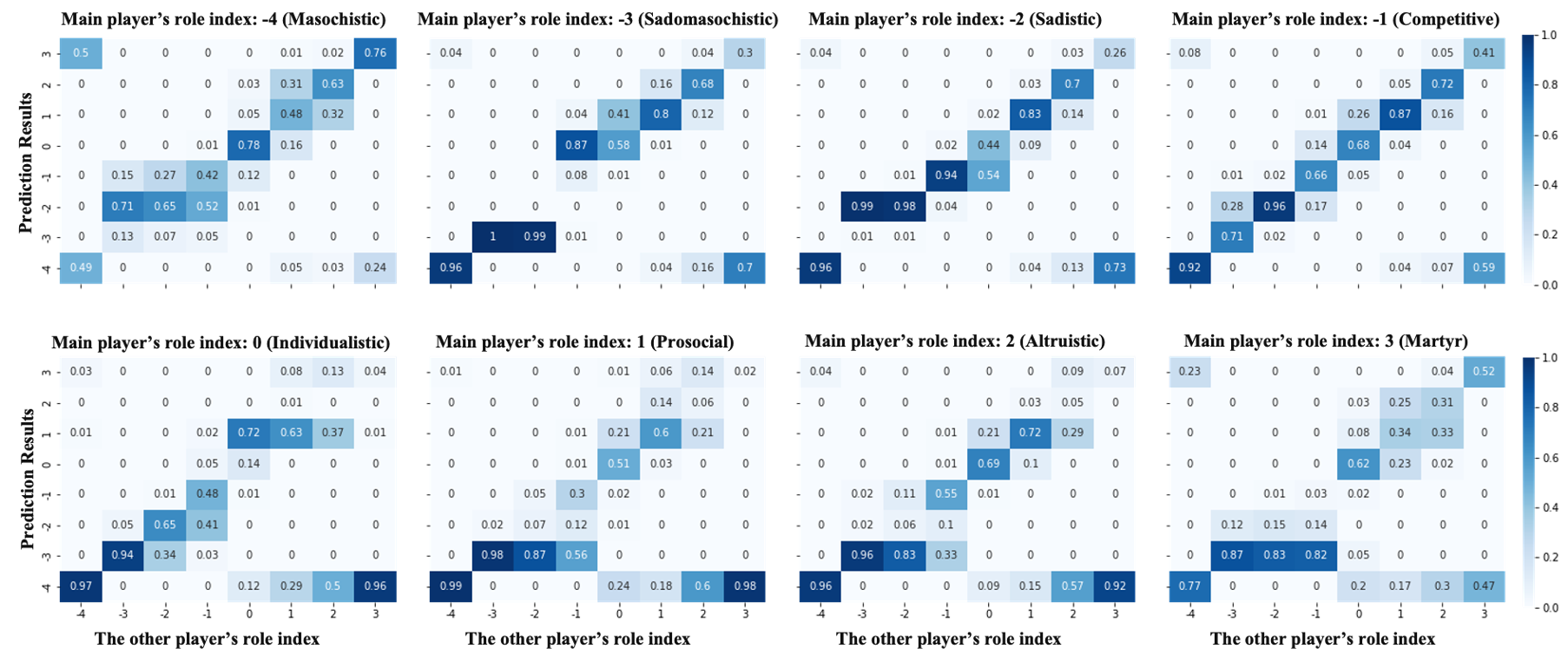}  
        \caption{Prediction results in Harvest}
        \label{fig_pred_harvest}
    \end{subfigure}
    \hfill
    \begin{subfigure}[b]{\textwidth}
        \centering
        \includegraphics[width=\textwidth]{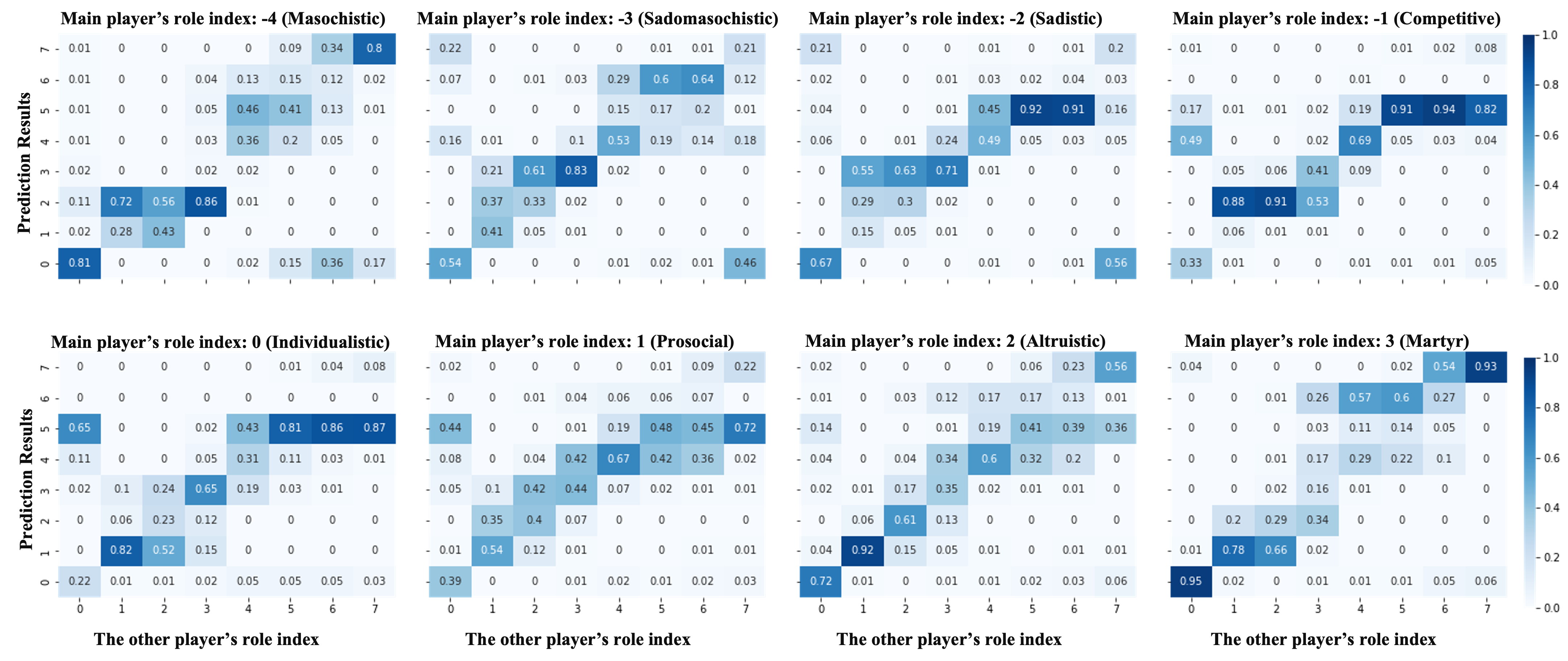}
        \caption{Prediction results in CleanUp}
        \label{fig_pred_cleanup}
    \end{subfigure}
    \caption{Prediction results of the role predictor $q(\phi)$.}
    \label{fig_pred}
\end{figure}

\section{Conclusion}
In this study, we introduced a novel framework called Role Play (RP) to enable agents to adapt varying roles, effectively addressing the zero-shot coordination challenge in MARL. Drawing inspiration from social intuition, RP incorporates a role predictor that estimates the joint role embeddings of other agents, which in turn facilitates the adaptation of the learning agent's role based on these predictions. We demonstrated the effectiveness of RP across both cooperative and mixed-motive environments, where agents must balance individual contributions against collective benefits, achieving comparable performance compared to state-of-the-art baselines.

However, the RP framework is not without its limitations. As the number of agents increases, the complexity of role prediction escalates, potentially challenging the ability of role predictor to accurately forecast the joint role embeddings of other agents. Additionally, the reward feature mapping function $\psi$ plays a crucial role in the framework, with its selection significantly influencing the learning process.

Future work will aim to tackle these challenges by enhancing the accuracy of role prediction and optimizing the choice of the reward feature mapping function. We also plan to extend the RP framework to more complex multi-agent scenarios, further testing its scalability and adaptability.

\medskip
{\small  
\bibliography{x}
} 


\appendix
\section{Theorem Proof}
\renewcommand{\thetheorem}{4.\arabic{theorem}}
\setcounter{theorem}{0}
\label{app1}
\begin{theorem}
    For a finite MDP with $T$ time steps and a specific role policy $\pi(z)$, if any random policy $\pi^\prime$ is $\epsilon$-close to the role policy $\pi(z^\prime)$, then we have
    \begin{equation}
        \left | \frac{J(\pi(z), \pi^\prime)}{J(\pi(z), \pi(z^\prime))} - 1 \right | \leq \epsilon T.
    \end{equation}
\end{theorem}
    
\begin{proof}
    We start by examining the trajectory probabilities under policies $\pi(z)$ and $\pi^\prime$:
    \begin{equation}
        p_{\pi(z), \pi^\prime}(\tau) = p(s_0) \prod_{t}^T \pi(a_t|o_t, z) \pi^\prime(a_t^\prime|o_t^\prime) p(s_{t+1}|s_t, a_t, a_t^\prime).
    \end{equation}
    Given the assumption  
    \begin{equation}
        \left | \frac{\pi^\prime(a_t^\prime|o_t^\prime)}{\pi(a_t^\prime|o_t^\prime, z)} - 1 \right | < \epsilon,
    \end{equation}
    we can approximate the ratio of trajectory probabilities under $\pi^\prime$ and $\pi(z^\prime)$:
    \begin{equation}
        \left | \frac{p_{\pi(z), \pi^\prime}(\tau)}{p_{\pi(z), \pi(z^\prime)}(\tau)} - 1 \right | \approx \left | \prod_{t=0}^T \frac{\pi^\prime(a_t^\prime|o_t^\prime)}{\pi(a_t^\prime|o_t^\prime, z)} - 1 \right | \approx \left | \left (\sum_{t=0}^T \frac{\pi^\prime(a_t^\prime|o_t^\prime)}{\pi(a_t^\prime|o_t^\prime, z)} - 1 \right ) \right |\leq \epsilon T.
    \end{equation}
    And the expected reward can then be written as:
    \begin{equation}
        J(\pi(z), \pi^\prime) = \sum_{\tau} p_{\pi(z), \pi^\prime}(\tau) \psi\left(\bm{\mathcal{R}}(\tau), z \right),
    \end{equation}
    where 
    \begin{equation}
        \bm{\mathcal{R}}(\tau) = \sum_{t}^T \bm{\mathcal{R}}(s_t, a_t, a_t^\prime),
    \end{equation}
    and $\psi$ is a predefined mapping function. Therefore, we have
    \begin{equation}
        \left | \frac{J(\pi(z), \pi^\prime)}{J(\pi(z), \pi(z^\prime))} - 1 \right | \leq \epsilon T.
    \end{equation}
\end{proof}

\section{Experimental Details}
\label{app2}
The baselines were established using the implementation details provided in the repositories from \citet{rahman2023generating} (AnyPlay, TrajeDi, BRDiv)\footnote{\url{https://github.com/uoe-agents/BRDiv}} and \citet{yu2023learning} (HSP)\footnote{\url{https://github.com/samjia2000/HSP}}. The network architecture for these baselines is identical to that of the RP framework, which is detailed below. The additional role predictor $q_\phi$ for RP is a one-layer fully connected neural network with 64 hidden units and ReLU activation function. We used the OpenAI gym implementation of Overcooked in HSP \citep{yu2023learning} and adapted it for other algorithms. The Harvest and CleanUp games were implemented using the MeltingPot framework \citep{leibo2021meltingpot}. Detailed experimental settings for Overcooked, Harvest, and CleanUp are provided in the following subsections. For more details, please refer to our code repository.\footnote{\url{https://github.com/Weifan408/role_play}}

\subsection{Experimental Details of Overcooked}
\label{app2_1}
The reward shaping function in Overcooked is defined as:
\begin{equation}
    \psi(r^i, z^i) = r^i + \sum_{k} (z^i_k * E_k),
\end{equation}
where $z_k^i \in [-1,0,1]$ is the preference of agent $i$ for event $k$, and $E_k$ is the reward associated with event $k$. The reward shaping function is designed to encourage agents to exhibit specific behaviors based on their role preferences. Table \ref{tab_2} shows the event reward we set in Overcooked. HSP uses the same reward shaping function as RP, while AnyPlay, TrajeDi, and BRDiv use the original collective reward directly from the environment.

\begin{table}[h]
    \centering
    \caption{Event reward in Overcooked.}
    \label{tab_2}
    \begin{tabular}{cc}
        \toprule
        Event & Reward \\
        \midrule
        Picking up an item from any dispenser (onion, tomato, dish) & 5 \\
        Picking up a soup & 5 \\
        Viable placement & 5 \\
        Optimal placement & 5 \\
        Catastrophic placement & 10 \\
        Placing an item into the pot (onion, tomato) & 3 \\
        Delivery & 10 \\
        \bottomrule
    \end{tabular}
\end{table}

Table \ref{tab_3} shows the detailed model architecture used in Overcooked. The policy and value networks are trained using the Proximal Policy Optimization (PPO) algorithm \citep{schulman2017proximal}.

\begin{table}[h]
    \centering
    \caption{Base model architecture used in Overcooked.}
    \label{tab_3}
    \begin{tabular}{cc}
        \toprule
        Layer & Architecture \\
        \midrule
        CNN Filter &  [[32,3,1],[64,3,1],[32,3,1]] \\
        CNN Activation & ReLU \\
        LSTM Cell Size & 128 \\
        LSTM Activation & SiLU \\
        Post Fcnet Hiddens & [64,64] \\
        Post Fcnet Activation & ReLU \\
        \bottomrule
    \end{tabular}
\end{table}

\subsection{Experimental Details in Harvest and CleanUp}
\label{app2_2}
We employ ring formulation of SVO to categorize agent roles. We introduced SVO-based (role-based) rewards as intrinsic motivators for agents to learn role-specific strategies. The reward feature mapping function $\psi$ is defined as:
\begin{equation}
    \psi(\bm{r}, z^i) = w * r^i + (1-w) * \underbrace{\left| \cos(z^i) r^i + \sin(z^i) \bar{r}^{-i} \right|}_{\text{SVO-based reward shaping}},
\end{equation}
where $r^i$ denotes the individual reward for agent $i$, $\bar{r}^{-i}$ the average reward of other agents, and $w$ a hyperparameter adjusting the impact of role-based rewards. In implementation, we set $w=0.3$, and the role embedding $z^i$ is one-hot encoded. Table \ref{tab_4} shows the hyperparameters used in Harvest and CleanUp.
\begin{table}[h]
    \centering
    \caption{Hyperparameters in Harvest and CleanUp.}
    \label{tab_4}
    \begin{tabular}{cc}
        \toprule
        Hyperparameters & Values \\
        \midrule
        FCnet Hiddens & [256,256] \\
        FCnet Activation & Tanh \\
        LSTM Cell Size & 256 \\
        LSTM Activation & SiLU \\
        w & 0.3 \\
        Trail Length & 10 \\
        \bottomrule
    \end{tabular}
\end{table}

We pre-trained three distinct policies using a naive SP framework and PPO for zero-shot evaluation: a selfish agent, a prosocial agent, and an inequity-averse agent \citep{hughes2018inequity}. Each agent is trained with a different reward function tailored to its role:
\begin{itemize}
    \item \textbf{Selfish Agent}: This agent is trained with the original reward directly from the environment.
    \item \textbf{Prosocial Agent}: This agent is trained using a collective reward that sums the rewards of all agents, formalized as $r_i = \sum_j r_j$.
    \item \textbf{Inequity-Averse Agent}: This agent's training includes an inequity shaping reward designed to address social fairness, defined by the formula: 
    \begin{equation*}
        r_i = r_i - \frac{\alpha\sum_{j\neq i}\max(0, r_j-r_i)+ \beta \sum_{j\neq i}\max(0, r_i-r_j)}{n-1},
    \end{equation*}
    where $\alpha\sum_{j\neq i}\max(0, r_j-r_i)$ and $\beta \sum_{j\neq i}\max(0, r_i-r_j)$ are the penalty for disadvantageous inequity and the reward for advantageous inequity, respectively. For this agent, $\alpha=5$ and $\beta=0.05$ are used,following the default values set in \citep{hughes2018inequity}.
\end{itemize}
Additionally, we modified the \textit{attack} action for both the selfish and inequity-averse agents to \textit{stay} in the CleanUp game to prevent them from becoming non-responsive (\textit{wooden man}) after training.

\end{document}